\documentclass{article}

\usepackage{graphicx,psfrag,epsf}
\usepackage{arxiv}
\usepackage[utf8]{inputenc} 
\usepackage[T1]{fontenc}    
\usepackage{hyperref}       
\usepackage{url}            
\usepackage{booktabs}       
\usepackage{amsfonts}       
\usepackage{nicefrac}       
\usepackage{microtype}      
\usepackage{lipsum}
\usepackage{natbib}
\usepackage{amsmath}
\usepackage{amssymb}
\usepackage{amsthm}
\def\T{{ \mathrm{\scriptscriptstyle T} }}

\def\uZ{{\bf Z}}
\def\uX{{\bf X}}
\def\uY{{\bf Y}}

\def\uD{{\bf D}}
\def\uV{{\bf V}}
\def\uI{{\bf I}}

\def\uS{{\bf S}}

\def\uF{{\bf F}}
\def\uc{{\bf c}}
\def\ub{{\bf b}}
\def\uu{{\bf u}}
\def\1{{\bf 1}}
\def\0{{\bf 0}}

\def\uomega{{\boldsymbol \omega}}
\def\uxi{{\boldsymbol \xi}}
\def\umu{{\boldsymbol \mu}}

\newtheorem{theorem}{Theorem}

\title{Synthetic control method with convex hull restrictions: A Bayesian maximum a posteriori approach}
\author{
  Gyuhyeong Goh\\
  Department of Statistics\\
  Kansas State University\\
  Manhattan, KS 66506\\
  \texttt{ggoh@ksu.edu} \\
   \And
  Jisang Yu\\
  Department of Agricultural Economics\\
  Kansas State University\\
  Manhattan, KS 66506\\
  \texttt{jisangyu@ksu.edu} \\
}

\begin{document}
\maketitle

\begin{abstract}
Synthetic control methods have gained popularity among causal studies with observational data, particularly when estimating the impacts of the interventions that are implemented to a small number of large units. Implementing the synthetic control methods faces two major challenges: a) estimating weights for each control unit to create a synthetic control and b) providing statistical inferences. To overcome these challenges, we propose a Bayesian framework that implements the synthetic control method with the parallelly shiftable convex hull and provides a useful Bayesian inference, which is drawn from the duality between a penalized least squares method and a Bayesian Maximum A Posteriori (MAP) approach. Simulation results indicate that the proposed method leads to smaller biases compared to alternatives. We apply our Bayesian method to the real data example of \cite{abadie2003economic} and find that the treatment effects are statistically significant during the subset of the post-treatment period.
\end{abstract}


\section{Introduction}
Unlike experimental data via random assignments of interventions, causal inferences of a policy effect from observational data face challenges due to confounding factors. In social science, including economics and political science, conducting experiments is often infeasible and thus, many researchers rely on observational data. The literature on causal inferences with observational data has been involving with the rapid development of empirical approaches that are based on the potential outcome framework \citep[e.g.][]{luo2019matching,athey2017state,holland1986statistics,rubin1974estimating}. Under the potential outcome framework, one can consistently estimate the ``average treatment effect" from the average difference between treated and untreated units conditional on the confounding factors, i.e. the factors that are correlated with potential outcomes and with treatments \citep{rosenbaum1983central}.

In this paper, we contribute to the growing literature of synthetic control methods by providing a novel Bayesian framework to estimate treatment effects and conduct causal inferences. Synthetic control methods, originally introduced by \cite{abadie2003economic} and further established by \cite{abadie2010synthetic}, are designed to estimate the impacts of the interventions that are implemented to a small number of large units such as regions or countries. As discussed by \cite{abadie2019}, employing traditional regression techniques, including the popular approaches like difference-in-differences or propensity score matching, is often inappropriate in such cases as those approaches do not lead to proper estimation of potential outcomes. One of the alternatives is to use time-series analyses for a single unit but it is challenging to control for confounding factors in traditional time-series approaches with few treated units. 

\cite{abadie2019,abadie2015comparative} and \cite{athey2017state} describe the strength of the synthetic control methods when conducting comparative case studies with the interest of estimating longer-term effects of infrequent policies. The synthetic control methods can estimate the over-time evolution in the potential outcomes of ``treated" units by constructing weighted sums of control units that mimic the evolution of the treated units during the pre-treatment period, i.e. the synthetic controls. Moreover, unlike in the difference-in-differences estimation, the synthetic control methods do not require the pre-trend parallel outcome assumption as normally they allow the outcomes of ``untreated" units to depend on unit-specific unobservables with time-varying coefficients.

Despite the strengths, the synthetic control methods face two major challenges: a) estimating weights for each control unit to create a synthetic control and b) providing statistical inferences. Our contribution to the literature is to provide a novel Bayesian approach that mitigates these two challenges. Using a Bayesian \emph{A Posteriori} (MAP) approach, we incorporate a \emph{parallelly shiftable} convex hull restriction to flexibly estimate the weights while preserving the nature of the convex hull restriction that prevents extrapolation biases. Also, we show that our approach naturally provides valid statistical inferences.

How to estimate the weights for constructing synthetic controls has been a great interest to econometricians and statisticians \citep[e.g.][]{doudchenko2016balancing,li2019statistical}. In the original framework of \cite{abadie2003economic} and \cite{abadie2010synthetic}, weights are estimated by a constrained optimization problem. The key feature of the constrained optimization problem is to employ the convex hull restriction. In other words, the optimal weights are the weights such that the estimated potential outcomes of the ``treated" unit are either in the convex hull of the predictors of the ``donors", i.e the ``untreated" units, or approximately close to the convex hull.  Recent development in the literature discusses relaxing the constraints that make the potential outcomes to be in or close to the convex hull \citep{li2019statistical,doudchenko2016balancing}. For example, \cite{doudchenko2016balancing} propose a regularization method that allows the weights to be negative, relaxing the constraint of the weights add up to one, and allows a permanent difference between the treated and the donors that is additive. \cite{li2019statistical} propose to modify the approach of \cite{abadie2015comparative} and \cite{abadie2010synthetic} by replacing the convex hull restriction for the weights with a conical hull restriction. 

However, these flexible approaches can increase the chances of extrapolation \citep{abadie2019}, which may lead to inconsistent estimations of treatment effects. The risk of extrapolation in estimating synthetic controls becomes more important as the numbers of potential ``donors" or predictors increase. Moreover, it is important to obtain sparse solutions on optimal weights to minimize the extrapolation biases and to obtain stable and interpretable counterfactuals \citep{abadie2019}. By showing that a \emph{parallelly shiftable} convex hull restriction also induces the sparsity in the weight estimation, we claim that our approach is flexible but still avoids the extrapolation biases.

As discussed by \cite{abadie2010synthetic}, standard inference techniques based on large samples are inappropriate for the treatment effects estimated by synthetic control methods. Alternative to the standard inference techniques, \cite{abadie2010synthetic} propose to estimate ``placebo effects" from a permutation distribution. However, permutation approaches often rely on distributional assumptions such as the exchangeability of the units \citep{firpo2018synthetic,hahn2017synthetic}. Recently, \cite{li2019statistical} and \cite{chernozhukov2018inference} propose different approaches that based on subsampling methods or k-fold cross-fitting not to rely on permutation distributions. We show that our Bayesian approach, which does not require data-partitioning nor subsampling, provides valid statistical inferences. 

The remainder of the paper is as follows. we first highlight the importance of preserving the convex hull restriction of \cite{abadie2003economic} and \cite{abadie2010synthetic} and propose a synthetic control method that uses a \emph{parallelly shiftable} convex hull. We also show that the convex hull restriction induces the sparsity in the estimation of the weights using Lagrangian duality. We then propose a novel Bayesian framework that implements the synthetic control method with the \emph{parallelly shiftable} convex hull and provides a useful Bayesian inference, which is drawn from the duality between a penalized least squares method and a Bayesian MAP approach. Finally, we conduct simulations to show that our proposed approach leads to smaller biases compared to alternative approaches and apply our proposed approach to the real data example of \cite{abadie2003economic}.

\section{Synthetic control method and convex hull restrictions}
For $i=1,\ldots, N$ and $t=1,\ldots, T$, let $Y_{it}^{(1)} \in \mathbb{R}$ be the potential outcome when unit $i$ is exposed to the treatment or intervention at time $t$ and $Y_{it}^{(0)}\in \mathbb{R}$ be the potential outcome of unit $i$ at time $t$ with no treatment. Without loss of generality, we assume that unit $1$ is exposed to the treatment over the last $T_1 (<T)$ time points, while the remaining $N-1$ units are unaffected by the treatment over the entire time period. The goal of causal studies is to estimate the treatment effect for unit $1$ at time $t$,  $\theta_{1t}=Y_{1t}^{(1)}-Y_{1t}^{(0)}$  for $t=T_0+1,\ldots, T$ or the average treatment effect,  $\bar{\theta}_{1}=T_1^{-1}\sum_{t=T_0+1}^T \theta_{1t}$, where $T_0=T-T_1$ is the length of the pre-treatment period. Given that $\{Y_{1t}^{(1)}, t>T_0\}$ is observed, estimating $\{\theta_{1t},t>T_0\}$ or $\bar{\theta}_{1}$ requires the estimation of $\{Y_{1t}^{(0)}, t>T_0\}$. 

The observed outcome of unit $i$ at time $t$ can be expressed as
$$Y_{it}=Y_{it}^{(0)}+\theta_{it} D_{it},$$
where $D_{it}=1$ if $i=1$ and $t>T_0$ and $D_{it}=0$ otherwise. Let $\uZ_{it}$ be a $p\times 1$ vector of covariates that are observed from unit $i$ at time $t$ and \emph{assumed to be unaffected by the treatment}. If $\uZ_{it}$'s are time-invariant, we define $\uZ_{i1}=\cdots=\uZ_{iT}$. The synthetic control estimator for the counterfactual potential outcomes, i.e. the potential outcomes of unit 1 with no treatment in time $t$, is defined as a linear combination of the untreated units (i.e. the donor pool), 
\begin{eqnarray}\label{scm:est1}
\hat{Y}_{1t}^{(0)}= \omega_1+\sum_{i=2}^{N} {\omega}_i Y_{it},\quad t=T_0+1,\ldots,T,
\end{eqnarray}
for some given $\omega_1,\ldots,\omega_{N}$. Let $\uomega=(\omega_1,\ldots,\omega_{N})^\T$ be the $N\times 1$ vector of \emph{unknown} coefficients (or weights) of the synthetic control estimator. For estimating $\uomega$, the main idea of the synthetic control methods is to solve the following system of equations with some constraints:
\begin{eqnarray}\label{eq:u1}
\begin{cases}
 &Y_{1t}= \omega_1+\sum_{i=2}^{N}{\omega}_i Y_{it} ,\quad t=1,\ldots,T_0,\\ 
 &\uZ_{1t}= \sum_{i=2}^{N}{\omega}_i \uZ_{it} ,\quad t=1,\ldots,T_0.
\end{cases}
\end{eqnarray}

\citet{abadie2003economic} and \citet{ abadie2010synthetic} propose to obtain $\uomega$ by solving Equation \eqref{eq:u1} with the ``convex hull constraint", that is,
\begin{eqnarray}\label{conv:const}
\uomega \in \mathcal{W}_{\text{conv}}=\left\{\uomega :\omega_1=0,
~\omega_2\geq 0,\ldots, \omega_N\geq 0 ,~\text{and}~ \sum_{i=2}^{N} \omega_i=1\right\}.
\end{eqnarray}
The solution of Equation \eqref{eq:u1} with this constraint exists only if $\uX_1=(\uY_1^\T , \uZ_{1}^\T)^\T$ is contained in the convex hull of $\{\uX_i=(\uY_{i}^\T, \uZ_{i}^\T)^\T:i=2,\ldots,N\}$, where $\uY_{i}=(Y_{i1},\ldots,Y_{iT_0})^\T$ and $\uZ_i=(\uZ_{i1}^\T,\ldots,\uZ_{iT_0}^\T)^\T$ for time-varying covariates (or $\uZ_i=\uZ_{i1}$ for time-invariant covariates). In practice, however, the solution often does not exist and \citet{abadie2003economic} and \citet{abadie2010synthetic} propose to obtain the \emph{approximate} solutions in such cases. The implementation of the synthetic control method of \citet{abadie2003economic} and \citet{abadie2010synthetic} is done by obtaining $\uomega$ that minimizes 
\begin{eqnarray}\label{object:f1} \left(\uX_1 - \sum_{i=2}^N \omega_i \uX_i\right)^\T \uV \left(\uX_1 - \sum_{i=2}^N \omega_i \uX_i\right) \quad \text{subject to $\uomega \in \mathcal{W}_{\text{conv}}$ }
\end{eqnarray}
for some positive definite matrix $\uV$, which allows for the \emph{approximate} solutions. 

Recent development in the literature proposes to relax the convex combination assumption for synthetic control estimators. For example, \citet{doudchenko2016balancing} propose to use the elastic net constraint \citep{zou2005regularization},
\begin{eqnarray}\label{enet:const}
\mathcal{W}_{\text{enet}}=\left\{\uomega \in \mathbb{R}^N:  \frac{1-\alpha}{2}\sum_{i=2}^N \omega_i^2 +\alpha \sum_{i=2}^N |\omega_i| \leq \lambda \right\}
\end{eqnarray}
for some $\alpha \in [0,1]$ and $\lambda >0$. Recently, \citet{li2019statistical} suggests to assume that $\uX_1$ belongs to the shifted conical hull of $\{\uX_i:i=2,\ldots,N\}$ by imposing the following constraint on $\uomega$,
\begin{eqnarray}\label{coni}
\mathcal{W}_{\text{coni}}=\left\{\uomega:\omega_1\in \mathbb{R}~ \text{and}~  \omega_i\geq 0,~i=2,\ldots,N\right\}.
\end{eqnarray}

It should be noted that removing the convex hull restriction is risky since it may lead to the extrapolation outside the range of the observed data \citep{abadie2015comparative, abadie2019}. 
In addition to that the convex hull restriction prevents extrapolation problems, another important feature is that the convex hull restriction can lead to a sparse solution to $\uomega$-estimation, i.e. many zero weights are assigned to the untreated units in the donor pool for constructing a synthetic control. Because of the sparsity, synthetic control methods with the convex hull restriction simultaneously execute both the estimation of the weights and the donor pool selection. While a geometric interpretation of this sparsity property is discussed by \citet{abadie2019}, a clear theoretical explanation has not yet been provided. In the following theorem, we provide a sufficient condition under which the convex hull restriction yields a sparse estimate of $\uomega$.
\begin{theorem}\label{thm:1} Define 
$$\mathcal{W}_1= \left\{\uomega :\omega_1=0
~\text{and}~ \sum_{i=2}^{N} \omega_i=1\right\}.$$
Let ${\uomega}^* =\arg\min_{{\uomega}  \in \mathcal{W}_1 } L(\uomega)$,
where 
\begin{eqnarray*}
L(\uomega)=\left(\uX_1 - \sum_{i=2}^N \omega_i \uX_i\right)^\T \uV \left(\uX_1 - \sum_{i=2}^N \omega_i \uX_i\right).
\end{eqnarray*}
Given $\uV$, if ${\uomega}^*\notin  \mathcal{W}_{\text{conv}}$, then the solution of the constrained optimization with the convex hull restriction in \eqref{object:f1} contains at least one zero entry (in addition to $\omega_1$).
\end{theorem}
\begin{proof}[Proof of Theorem \ref{thm:1}] Let $\hat{\uomega}$ be the solution of \eqref{object:f1}. By the Lagrangian duality, $\hat{\uomega}$ is obtained by minimizing
\begin{eqnarray}\label{L:eq}
\min_{{\uomega} ,\lambda ,\umu:\omega_1=0} \left\{ L(\uomega)+\lambda\left(\sum_{i=2}^n\omega_i -1\right)- \sum_{i=2}^n\mu_i\omega_i \right\}.
\end{eqnarray}
By the Karush-Kuhn-Tucker (KKT) theorem, the following conditions must hold for $\hat{\uomega}$:
\begin{eqnarray*}
&\mu_i \geq 0 \quad \text{for all}~i=2,\ldots,N,~&\text{(dual feasibility)} \\
&\mu_i \hat{\omega}_i =0 \quad \text{for all} ~i=2,\ldots,N,~&\text{(complementary slackness)}\\
&\sum_{i=2}^N \hat{\omega}=1 \quad\text{and}\quad \hat{\omega}_i \geq 0 \quad \text{for all} ~i=2,\ldots,N.~ &\text{(primal feasibility)}
\end{eqnarray*}
By the dual feasibility and the complementary slackness conditions, we must have either (i) $\hat{\omega}_i=0$ for some $i \in \{2,\ldots,N\}$ or (ii) $\mu_i=0$ for all $i\in \{2,\ldots,N\}$. Suppose that the second situation occurs, i.e. $\mu_i=0$ for all $i\in \{2,\ldots,N\}$. Then, the optimization problem in \eqref{L:eq} can be simplified as 
\begin{eqnarray*}
\min_{{\uomega} ,\lambda :\omega_1=0} \left\{ L(\uomega)+\lambda\left(\sum_{i=2}^n\omega_i -1\right)\right\}.
\end{eqnarray*}
This implies that $\hat{\uomega}=\uomega^*$, namely $\uomega^* \in \mathcal{W}_{\text{conv}}$, which contradicts our assumption that $\uomega^* \notin \mathcal{W}_{\text{conv}}$. Therefore, if $\uomega^* \notin \mathcal{W}_{\text{conv}}$, then at least one of $\hat{\omega}_i$ $(i=2,\ldots,N)$ must be zero.
\end{proof}

Theorem \ref{thm:1} formalizes the sparsity discussion of \cite{abadie2019}. In many practical applications of the synthetic control methods, the sufficient condition described in Theorem \ref{thm:1} can hold since $\mathcal{W}_{\text{conv}}$ is a much smaller subset of $\mathcal{W}_1$. Hence, synthetic control methods with the convex hull restriction induces the sparsity in most cases. 

However, as discussed by \cite{doudchenko2016balancing}, the convex hull restriction prevents to achieve the balance between the treated unit and the synthetic control if there is a permanent additive difference between the outcomes of the treated unit and the outcomes of the control units. Standard difference-in-differences argument is that we can ``difference out" such permanent additive differences. Therefore, allowing for such differences in synthetic control methods can improve the estimation performance.

To avoid extrapolation and enhance the flexibility of the synthetic control methods simultaneously, we propose to allow \emph{parallel shifts} of convex hull constraints. Formally, this can be defined by
\begin{eqnarray}\label{lsconv:const}
\mathcal{W}_{\text{ps-conv}}=\left\{\uomega:\omega_1\in \mathbb{R},~ ~\omega_2\geq 0,\ldots, \omega_N\geq 0 ,~\text{and}~\sum_{i=2}^N \omega_i = 1 \right\}.
\end{eqnarray}
Let 
$${\uX}_{0}=\left[\begin{matrix}\1 &\uY_2&\cdots&\uY_N \\
\0&\uZ_2&\cdots&\uZ_N
\end{matrix}\right],$$
where $\1$ denotes a column vector of ones and $\0$ denotes a column vector of zeros. Then, $\uomega$ can be obtained by minimizing
\begin{eqnarray}\label{con:min}
 \min_{\uomega \in \mathcal{W}_{\text{ps-conv}}} \left(\uX_1 -  \uX_0\uomega \right)^\T \uV \left(\uX_1 -  \uX_0\uomega \right),
\end{eqnarray}
for given $\uV$, where the specific structure of $\uV$ is given in Section \ref{sec:4}. 

The following theorem shows that the sparsity in the weight estimation of the synthetic control methods is still preserved with the parallelly shiftable convex hull constraint.
\begin{theorem}\label{thm:2} 
Let 
\begin{eqnarray*} \uomega^{**}=\arg\min_{\uomega\in \mathcal{W}_2 } \left(\uX_1 -  {\uX}_0\uomega \right)^\T \uV \left(\uX_1 -  {\uX}_0\uomega \right),
\end{eqnarray*}
where 
\begin{eqnarray*}
\mathcal{W}_2=\left\{\uomega :\omega_1\in \mathbb{R},
~\sum_{i=2}^N \omega_i=1 \right\}. 
\end{eqnarray*}
Given $\uV$, if ${\uomega}^{**} \notin  \mathcal{W}_{\text{ps-conv}}$, then the solution of the constrained optimization with the convex hull restriction in \eqref{con:min} contains at least one zero entry.
\end{theorem}
\begin{proof}[Proof of Theorem \ref{thm:2}]
The proof is analogous to the proof of Theorem \ref{thm:1}.
\end{proof}

Let $\hat{\uomega}$ be the solution of \eqref{con:min}. We then define a new synthetic control estimator as
\begin{eqnarray}\label{sscm:est}
\hat{Y}_{1t}^{(0)}= \hat{\omega}_1+\sum_{i=2}^{N} \hat{\omega}_i Y_{it},\quad t=T_0+1,\ldots,T.
\end{eqnarray}
The proposed synthetic control estimate can be sensitive to the choice of the tuning parameter $\uV$ when the solution of Equation \eqref{eq:u1} does not exist subject to $\omega \in \mathcal{W}_{\text{ps-conv}}$. For tuning parameter selection, cross-validation techniques have been commonly employed. However, the synthetic control estimates resulting from cross-validations can vary with different data-splitting schemes. Another statistical challenge is conducting statistical inferences of the synthetic control estimators. To address the aforementioned issues, we propose a Bayesian inference method that draws \emph{posterior probabilistic} conclusions about the synthetic control estimator given the observed data. In a Bayesian framework, the uncertainty associated with tuning parameter selection can be integrated out by using the marginal posterior distribution of $\uomega$.

\section{Bayesian synthetic control methods}\label{sec:4}
A Maximum \emph{A Posteriori} (MAP) approach, which utilizes the mode of the posterior distribution as a point estimate, has been used to facilitate Bayesian inferences without assuming the distribution of data  \citep[e.g.][]{grendar2009asymptotic, goh2018bayesian, florens2019gaussian}. Also, from a Bayesian perspective, a penalized least squares estimate is equivalent to a MAP estimate under some shrinkage prior--often referred to as Bayesian duality \citep{polson2016mixtures, bhadra2019lasso}. For example, the lasso estimate can be interpreted as the MAP estimate of Bayesian linear regression with the double-exponential (or Laplace) prior distribution \citep{tibshirani1996regression, wang2007robust, park2008bayesian}.

From the duality property, the constrained optimization in \eqref{con:min} is equivalent to the MAP estimation of the following Bayesian model:
\begin{eqnarray}
\label{like:X} f(\uX_1|\uomega,\uV)&\propto& \exp\left\{-\frac{1}{2}\left\|\uV^{1/2}\left(\uX_1  -  \uX_0\uomega \right)\right\|^2\right\},\\
\label{prior:w} \pi(\uomega)&\propto& \mathbb{I}\left(\sum_{i=2}^N \omega_i=1\right) \prod_{i=2}^N \mathbb{I}(\omega_i\geq0),
\end{eqnarray}
where $f(\uX_1|\uomega,\uV)$ represents the (pseudo) likelihood function and $\pi(\uomega)$ denotes the prior density function of $\uomega$. Note that the duality follows from Bayes theorem, $\pi(\uomega |\uX_1,\uV)\propto f(\uX_1|\uomega,\uV)\pi(\uomega)$. However, the Bayesian duality property is invalid if the resulting posterior distribution is improper. Since $\pi(\uomega)$ is an improper prior, it is crucial to check the propriety of our posterior distribution. The following theorem shows that the posterior distribution is proper.
\begin{theorem}\label{thm:3}
The proposed Bayesian method produces a proper posterior distribution, that is,
$$ \int f(\uX_1|\uomega,\uV)\pi(\uomega) d\uomega <\infty.$$
\end{theorem}
\begin{proof}[Proof of Theorem \ref{thm:3}] 
Let $\ell(\uomega)=\|\uV^{1/2}\left(\uX_1  -\uX_{0}\uomega \right)\|^2$ and $\tilde{\uomega}=\arg\min_{\uomega\in \mathbb{R}^N }\ell(\uomega)$. Then, we have
\begin{eqnarray*}
 \ell({\uomega})&\geq&   \left\|\uV^{1/2}\left(\tilde{\uX}_1 - \omega_1\uX_{01} \right)\right\|^2\\
 &=&\left\|\uV^{1/2}\left(\tilde{\uX}_1 - \tilde{\omega}_1\uX_{01} \right)\right\|^2+ (\uX_{01}^\T \uV\uX_{01})(\omega_1 -\tilde{\omega}_1)^2
\end{eqnarray*}
for any $\omega_1 \in \mathbb{R}$, where $\uX_{01}$ indicates the first column of $\uX_0$ and $\tilde{\uX}_1=\uX_1-\sum_{i=2}^N \tilde{\omega}_i \uX_i$. 
This implies that
\begin{eqnarray*}
\nonumber \int f(\uX_1|\uomega,\uV)\pi(\uomega) d\uomega &\leq& c \int \exp\left\{-\frac{1}{2}(\uX_{01}^\T \uV\uX_{01})(\omega_1 -\tilde{\omega}_1)^2\right\}d\omega_1\\
&&\times \int \cdots \int \mathbb{I}\left(\sum_{i=2}^N \omega_i=1\right) \prod_{i=2}^N \mathbb{I}(\omega_i\geq0) d\omega_2 \cdots d\omega_N \label{eq:finite}
\end{eqnarray*}
for some positive constant $c$. From the normal density function and the Dirichlet density function, it is immediate that the integrals are finite.
\end{proof}

In a Bayesian framework, to select the optimal values of $\uV$, we can assign prior distributions for $\uV$. As in the literature on synthetic control methods including \cite{abadie2003economic} and \cite{abadie2010synthetic}, we assume that $\uV$ is a diagonal matrix. Specifically, we consider the following form of $\uV$:
\begin{eqnarray}\label{a:V}
\uV=\frac{1}{\nu}  \text{Diag}( \uI_{T_0},\xi_1 \uI_{R_1} ,\ldots,\xi_p \uI_{R_p}),
\end{eqnarray}
where $\nu (>0)$ is a scale parameter that measures the loss of the quadratic approximation to the exact estimating equation \eqref{eq:u1} , $\xi_j (>0)$ controls the influence  of the $j$-th covariate on the synthetic control estimator, and $R_j$ indicates the dimension of the $j$-th covariate. 

Under the proposed structure \eqref{a:V}, we assign the following conjugate priors for $\nu$:
\begin{eqnarray*}
\pi(\nu)&\propto& \nu^{-c_0-1}\exp(-d_0/\nu)\mathbb{I}(\nu>0),
\end{eqnarray*}
where $c_0(>0)$ and $d_0(>0)$ are prespecified hyperparameters. To incorporate automatic variable selection into our posterior inference, we consider
$$\xi_j\sim \text{Ber}(\eta_{0j}),\quad j=1,\ldots,p,$$
where $\eta_{0j}\in[0,1]$ $(j=1,\ldots,p)$ are prespecified hyperprarmeters. In this paper, we set $c_0=d_0=0.5$ and $\eta_{01}=\cdots=\eta_{0p}=0.5$, so that the priors are approximately non-informative.

The MAP estimate of $\uomega$ is obtained by solving $\max_{\uomega} \pi(\uomega | \uX_1,\uX_0)$ or equivalently solving $\max_{\uomega} f(\uX_1|\uomega)\pi(\uomega)$, where 
$$  f(\uX_1|\uomega) =  \int  f(\uX_1|\uomega,\nu,\uxi) \pi(\nu)\pi(\uxi) d\nu d\uxi.$$
However, it is challenging to obtain a closed-form expression for $ f(\uX_1|\uomega)$. To overcome this challenge, we propose to use the Monte Carlo Expectation Maximization (EM) algorithm of \citet{wei1990monte} which proceeds by iterating the following steps:
\begin{itemize}
    \item[1.] (S-step) Generate a sample $\{\uxi^{(m)},\nu^{(m)}:m=1,\ldots,M\}$ from $\pi(\uxi,\nu |\uX_1,\uX_0,\uomega^{t})$.
    \item[2.] (E-step) Compute 
    $$Q(\uomega| \uomega^{t})=M^{-1} \sum_{m=1}^M \log f(\uX_1|\uomega,\nu^{(m)},\uxi^{(m)}).$$
    \item[3.] (M-step) Update
    $$\uomega^{t+1} = \arg\max_{\uomega \in \mathcal{W}_{\text{ps-conv}} } Q(\uomega| \uomega^{t}). $$
\end{itemize}

To implement S-Step, we use the Gibbs sampler that iteratively generates a sample from $\pi(\nu|\uxi,\uX_1,\uX_0,\uomega^{t})$ and $\pi(\uxi|\nu,\uX_1,\uX_0,\uomega^{t})$. The explicit full conditional posterior distributions are given in Section \ref{sec:Gibbs}. In E-Step, after some algebraic manipulations, we have
$$Q(\uomega | \uomega^t)= \frac{1}{2}\left(\uX_1 -  \uX_0\uomega \right)^\T \bar{\uV}^t \left(\uX_1 -  \uX_0\uomega \right) +\text{constant},$$
where
$$\bar{\uV}^t=\left(M^{-1}\sum_{m=1}^M 1/\nu^{(m)}\right)  \text{Diag}\left\{ \uI_{T_0},\left( \sum_{m=1}^M\xi_1^{(m)}/M\right)\uI_{R_1},\ldots,\left(\sum_{m=1}^M\xi_p^{(m)}/M\right)\uI_{R_p}\right\}.$$

It is worth noting that, by Theorem \ref{thm:2}, M-Step produces a sparse solution for the MAP estimate. 
As a result, a valid donor pool can be immediately constructed by selecting non-zero elements of the MAP estimate. Let $\hat{\uomega}$ be the MAP estimate obtained from the Monte Carlo EM algorithm. Using the sparsity of $\hat{\uomega}$, we can divide the donor pool into two parts: $\mathcal{A}=\{i=2,\ldots,N:\hat{\omega}_i\neq 0\}$ and $\mathcal{A}^c=\{i=2,\ldots,N:\hat{\omega}_i= 0\}$, where $\mathcal{A}$ is the set of valid donors and $\mathcal{A}^c$ consists of invalid donors. Given the valid donor pool $\mathcal{A}$, the prior of $\uomega$ is defined as
\begin{eqnarray}
\label{prior:w1} \pi(\uomega|\mathcal{A})&\propto& \mathbb{I}\left(\sum_{i\in \mathcal{A} } \omega_i=1\right) \left\{\prod_{i\in \mathcal{A}} \mathbb{I}(\omega_i\geq0)\right\}\left\{\prod_{i\in \mathcal{A}^c} \mathbb{I}(\omega_i=0)\right\}.
\end{eqnarray}

Hence, Bayesian inferences about $\omega$ given $\mathcal{A}$ can be performed by drawing a sample from $\pi(\uomega|\uX_1,\uX_0,\mathcal{A})$. Let $\{\uomega^{(m)}:m=1,\ldots,M\}$ be a posterior sample generated from $\pi(\uomega|\uX_1,\uX_0,\mathcal{A})$. Then, we obtain Bayesian inferences for the treatment effects by using \begin{eqnarray*}
\theta_{1t}^{(m)}=Y_{1t}-\left( {\omega}_{1}^{(m)}+\sum_{i=2}^{N} {\omega}_{i}^{(m)} Y_{it}\right),\quad t=T_0+1,\ldots,T,~m=1,\ldots,M,
\end{eqnarray*}
since $\{({\omega}_{1T_0+1}^{(m)},\ldots,{\omega}_{1T}^{(m)}):m=1,\ldots,M\}$ is generated from the joint posterior distribution,  $\pi({\omega}_{1T_0+1},\ldots,{\omega}_{1T}|\uX_1,\uX_0,\mathcal{A})$.

However, to generate a sample from $\pi(\uomega|\uX_1,\uX_0,\mathcal{A})$, we face the following intractable integration problem:
$$\pi(\uomega|\uX_1,\uX_0,\mathcal{A})=\int \int  \pi(\uomega, \nu, \uxi |\uX_1,\uX_0,\mathcal{A}) d\nu d\uxi.$$
Due to the use of the conjugate hyperpriors, the above algebraic difficulty can be avoided by using the Gibbs sampler \citep{gelfand1990sampling}. The implementation of Gibbs sampling is discussed in the following section. 

\section{Posterior inference using Gibbs sampling}\label{sec:Gibbs}
As mentioned in the previous section, the practical issue of implementing the proposed Bayesian inference procedure is how to generate a sample from the posterior distribution of $\uomega$ given the data and the valid donor pool, $\pi(\uomega|\uX_1,\uX_0,\mathcal{A})$. To address this challenge, we use a Markov Chain Monte Carlo (MCMC) method with the Gibbs sampler. In a hierarchical Bayesian framework, the Gibbs sampling algorithm generates a Gibbs sequence of random variables
$$\uomega^{(0)}, \nu^{(0)}, \uxi^{(0)} ,\ldots, \uomega^{(m)},\nu^{(m)}, \uxi^{(m)} , \ldots $$
by iteratively sampling from the following full conditional distributions: 
\begin{eqnarray*}
\uomega^{(m+1)}&\sim& \pi(\uomega| \nu^{(m)}, \uxi^{(m)} , \uX_1,\uX_0,\mathcal{A} ),\\
\nu^{(m+1)}&\sim& \pi(\nu|\uomega^{(m+1)}, \uxi^{(m)} , \uX_1,\uX_0,\mathcal{A}),\\
\uxi^{(m+1)}&\sim& \pi(\uxi|\uomega^{(m+1)}, \nu^{(m+1)}, \uX_1,\uX_0,\mathcal{A}).
\end{eqnarray*}

We now derive the exact full conditional distribution for each parameter. For $\uomega$, the full conditional posterior density is proportional to
\begin{eqnarray*}
\pi(\uomega| \nu, \uxi , \uX_1,\uX_0,\mathcal{A}  ) &\propto&  f(\uX_1|\uomega,\nu,\uxi) \pi(\uomega|\mathcal{A}) \\
&\propto& \exp\left\{-\frac{1}{2}\left\|\uV^{1/2}\left(\uX_1  -  \uX_0\uomega \right)\right\|^2\right\} \\
&&\times \mathbb{I}\left(\sum_{i\in \mathcal{A} } \omega_i=1\right) \left\{\prod_{i\in \mathcal{A}} \mathbb{I}(\omega_i\geq0)\right\}\left\{\prod_{i\in \mathcal{A}^c} \mathbb{I}(\omega_i=0)\right\}.
\end{eqnarray*}
While sampling jointly from $\pi(\uomega| \nu, \uxi , \uX_1,\uX_0,\mathcal{A}  )$ is difficult, the component-wise Gibbs sampling, i.e. generating $\omega_j\sim \pi(\omega_j|\uomega_{-j},\nu, \uxi , \uX_1,\uX_0,\mathcal{A} ),$ is available as follows: 
\begin{itemize}
\item Generate $\omega_1 \sim \text{N}(\mu_1,\sigma_1^2)$, where
$ \mu_1=(\uX_{01}^\T \uV\uX_{01})^{-1} \uX_{01}^\T \uV (\uX_1- \sum_{i\in \mathcal{A}}\omega_{i} 
\uX_i)$ and $\sigma^2_1=(\uX_{01}^\T \uV\uX_{01})^{-1}$; 
\item For $j \in \mathcal{A}/\{ \max (\mathcal{A}) \}$, 
generate $\omega_i\sim\text{TN}_{[0, U_j]} (\mu_j,\sigma_j^2)$, where 
$ \mu_j=(\uX_{j}^\T \uV\uX_{j})^{-1} \uX_{j}^\T \uV (\uX_1- \omega_1\uX_{01}- \sum_{i\in \mathcal{A}\setminus\{j\} }\omega_{i} 
\uX_i)$, $\sigma^2_j=(\uX_{j}^\T \uV\uX_{j})^{-1}$, and $U_i=1-\sum_{j\in \mathcal{A}\setminus\{i\} }\omega_j$; 
\item For $j=\max (\mathcal{A})$, set $\omega_j = 1-\sum_{j \in \mathcal{A}/\{ \max (\mathcal{A}) \} } \omega_i$;
\item For $j \in \mathcal{A}^c$ set $\omega_j=0$.
\end{itemize}

Since $\pi(\nu)$ is conjugate to the likelihood function, it follows from Bayes' theorem that
\begin{eqnarray}
\nonumber &&\nu|\uomega, \uxi ,\uX_1,\uX_0,\mathcal{A}  \\
\label{nu:full}&&\quad \sim \text{Inverse-Gamma}\left( \frac{T_0+\sum_{j=1}^p \xi_j R_j}{2}+c_0, \frac{\left\|\uS_\uxi^{1/2}\left(\uX_1 - \uX_0\uomega\right)\right\|^2}{2} +d_0 \right).
\end{eqnarray}
where $\uS_\uxi^{1/2}=\text{Diag}( \uI_{T_0},\sqrt{\xi_1} \uI_{R_1} ,\ldots,\sqrt{\xi_p} \uI_{R_p})$.
Similarly, from Bayes' theorem, we have
\begin{eqnarray*}
\pi(\xi_j| \uomega, \nu,\uX_1,\uX_0,\mathcal{A}  )\propto \phi\left(\left.\uD^j_{\xi_j} {\uZ}_{1}  \right| \uD^j_{\xi_j}\uZ_{0} \uomega, \nu \uI \right) \eta_{0j}^{\xi_j}(1-\eta_{0j})^{\xi_j},\quad j=1,\ldots,p,
\end{eqnarray*}
where $\uD^j_{\xi}=\text{Diag}(0 \uI_{R_1} ,\ldots, 0 \uI_{R_{j-1}} , \xi \uI_{R_j} , 0 \uI_{R_{j+1}} ,\ldots, 0 \uI_{R_{p}}  )$ and $\phi(\cdot|\mu,\Sigma)$ denotes a multivariate normal density function with mean vector $\mu$ and covariance matrix $\Sigma$.  This leads to
\begin{eqnarray}
\label{xi:full}\xi_j | \uomega, \nu,\uX_1,\uX_0,\mathcal{A} \overset{ind}{\sim} \text{Bernoulli}\left[\left\{1+\frac{(1-\eta_{0j})\phi\left(\uD^j_{0} {\uZ}_{1}  | \uD^j_{0}\uZ_{0} \uomega, \nu \uI \right)}{\eta_{0j}\phi\left(\uD^j_{1} {\uZ}_{1}  |\uD^j_{1} \uZ_{0} \uomega, \nu \uI \right) }\right\}^{-1}\right]
\end{eqnarray}
for $j=1,\ldots,p$. 

Note that $\nu$ and $\mathcal{A}$ are conditionally independent given $(\uomega, \uxi ,\uX_1,\uX_0)$. Similarly, $\uxi$ and $\mathcal{A}$ are conditionally independent given $(\uomega, \nu ,\uX_1,\uX_0)$. Hence, \eqref{nu:full} and \eqref{xi:full} can be directly used for generating a sample from $\pi(\nu|\uxi,\uX_1,\uX_0,\uomega)$ and $\pi(\uxi|\nu,\uX_1,\uX_0,\uomega)$ in the Monte Carlo EM algorithm discussed in Section \ref{sec:4}.
\section{Simulation study}
In this section, we examine the finite sample performance of the proposed synthetic control method using a Monte Carlo simulation study. Given $N=40$, $T=100$, $T_0=40$, and $p=8$, we simulate the data $\{(Y_{it},\uZ_i):i=1,\ldots, N,t=1,\ldots,T\}$ from the following linear factor model: 
\begin{eqnarray*}
Y_{it}&=& Y_{it}^{(0)}+\theta_0 \left(0.5+\sqrt{t/2}\right)D_{it},\\
Y_{it}^{(0)}&=&\mu_i+\delta_t+\uc_t^\T \uZ_i+\ub_{i}^\T \uF_t+\epsilon_{it},\\
\uF_t &=&0.2 \uF_{t-1}+\uu_{t},
\end{eqnarray*}
where $D_{it}=\mathbb{I}(i=1,t>T_0)$, $\mu_i\overset{iid}{\sim} \text{Uniform}(-1,1)$, $\delta_t=\sqrt{5 t}$, $\uZ_i \sim N_p(\1_{p},2\uI_{p})$, $\uc_t=(c_{t1},\ldots,c_{tp})^\T$ such that $c_{t1},c_{t2}\overset{iid}{\sim} \text{Uniform}(-0.2,0.2)$ and $c_{t3}=\cdots=c_{tp}=0$, $\ub_i\sim N_3(\0_3,0.5\uI_3)$, $\uF_0 \sim N_3(0,\uI_3)$,  $\uu_t \overset{iid}{\sim} N_3(0,0.25\uI_3)$, $\epsilon_{it}\overset{iid}{\sim} N(0,0.1)$, and varying values are considered for $\theta_0$. 

In the data generating process, the treatment effect of unit $1$ at time $t$ during the post-treatment period is given by
$$  \theta_{1t}=\theta_0 \left(0.5+\sqrt{t/2}\right).$$
Hence, if we set $\theta_0>0$ $(\theta_0<0)$, then the treatment effect is increasing (decreasing) with time $t$ for $t>T_0$. When we set $\theta_0=0$, there is no treatment effect for all $t$. In this simulation study, we consider the following three cases: $\theta_0=-1$, $\theta_0=0$, and $\theta_0=1$. Also, note that $\uZ_i$ is the vector of $p$ time-invariant covariates and that the first two covariates are relevant to the outcome but the remaining $p-2$ are irreverent. 

To demonstrate the role of different constraints and to compare the magnitude of bias among alternative methods, we compare the performances of the synthetic control estimators with the following four constraints: 1) $\uomega \in  \mathcal{W}_{\text{conv}} $ \citep{abadie2003economic,abadie2010synthetic}, 2) $\uomega \in \mathcal{W}_{\text{coni}}$ \citep{li2019statistical}, 3) $\uomega \in \mathcal{W}_{\text{enet}}$ \citep{doudchenko2016balancing}, and 4) $\uomega \in \mathcal{W}_{\text{ps-conv}}$ (Proposed), where $\mathcal{W}_{\text{conv}}$, $\mathcal{W}_{\text{enet}}$, $\mathcal{W}_{\text{coni}}$, and $\mathcal{W}_{\text{ps-conv}}$ are defined in Equations \eqref{conv:const}, \eqref{enet:const}, \eqref{coni},  and \eqref{lsconv:const}, respectively. 

First, we obtain an estimate of $\uomega$ by minimizing $\|\uX_1-\uX_0\uomega\|^2$ subject to $\uomega \in \mathcal{W}_{\text{conv}}$ (labeled as ADH-SCM). Second, following \citet{doudchenko2016balancing}, we compute an estimate for $\uomega$ by minimizing $\|\uY_1-\uY_0\uomega\|^2$ subject to $\uomega \in \mathcal{W}_{\text{enet}}$ (labeled as DI-SCM), where $\uY_0=[\1,\uY_2,\ldots,\uY_N]$ and the tuning parameters of elastic net penalty, $\alpha$ and $\lambda$, are selected by 10-fold cross-validation using \texttt{cva.glmnet} function in \texttt{R} package \texttt{glmnetUtils}. Third, following \citet{li2019statistical}, we obtain an estimate of $\uomega$ by minimizing $\|\uY_1-\uY_0\uomega\|^2$ subject to $\uomega \in \mathcal{W}_{\text{coni}}$ (labeled as L-SCM). Fourth, we obtain an estimate of $\uomega$ by minimizing $\|\uX_1-\uX_0\uomega\|^2$ subject to $\uomega \in \mathcal{W}_{\text{ps-conv}}$ (labeled as na\"ive Bayes SCM). Lastly, we obtain the proposed MAP estimate of $\uomega$ using the Monte Carlo EM algorithm (labeled as Bayes SCM).

To evaluate the performance of estimation, we consider the following two mean squared errors (MSEs): 
\begin{eqnarray*}
\text{MSE}_{\text{TE}}= \frac{1}{T_1}\sum_{t=T_0+1}^{T} E \left\{\left(\hat{\theta}_{1t}-\theta_{1t}\right)^2\right\} , \quad 
\text{MSE}_{\text{ATE}}=E\left\{\left(\hat{\bar{\theta}}_{1}-\bar{\theta}_{1}\right)^2\right\},
\end{eqnarray*}
where $\hat{\theta}_{1t}=Y_{1t}-\hat{\omega}_1-\sum_{i=2}^N\hat{\omega}_i Y_{it}$, $\hat{\bar{\theta}}_{1}=\sum_{t=T_0+1}^{T} \hat{\theta}_{1t}/T$, and $\hat{\uomega}=(\hat{\omega}_1,\ldots,\hat{\omega}_N)$ is the estimate of $\uomega$ obtained by each method. $\text{MSE}_{\text{TE}}$ and $\text{MSE}_{\text{ATE}}$ measure the performance of each method in estimating the treatment effects and the average treatment effect, respectively. 

The simulation results are reported in Table \ref{Table:1}. The proposed method, Bayes SCM, has the smallest MSE among the five synthetic control methods. By comparing na\"ive Bayes SCM with ADH-SCM, we can see the benefit the proposed ``parallelly shiftable convex hull'' constraint. This comparison shows that allowing a parallel shift of the convex hull leads to substantial reduction in the MSEs (about 18 \% decrease in average $\text{MSE}_{\text{TE}}$ and about 70\% decrease in average $\text{MSE}_{\text{ATE}}$). By comparing na\"ive Bayes SCM with L-SCM, we also observe that removing the condition that the weights need to add up to one significantly inflates MSEs. Finally, by comparing Bayes SCM with na\"ive Bayes SCM, we can see the merit of the proposed Bayesian approach for tuning parameter selection problems as we observe about 17 \% decrease in average $\text{MSE}_{\text{TE}}$ and about 20\% decrease in average $\text{MSE}_{\text{ATE}}$ by integrating the tuning parameters out in a Bayesian fashion. 

\begin{table}[ht]
\centering
\caption{\label{Table:1} Summary of simulation study (MSEs are estimated by using $1,000$ Monte Carlo replications.)}
\begin{tabular}{r|cc|cc|cc}
  \hline
 & \multicolumn{2}{c|}{$\theta_0=1$} &\multicolumn{2}{c|}{$\theta_0=0$} &\multicolumn{2}{c}{$\theta_0=-1$}\\
 Method &  $\text{MSE}_{\text{TE}}$ & $\text{MSE}_{\text{ATE}}$ & $\text{MSE}_{\text{TE}}$ & $\text{MSE}_{\text{ATE}}$ & $\text{MSE}_{\text{TE}}$ &$\text{MSE}_{\text{ATE}}$  \\ 
  \hline
ADH-SCM& 0.1354 &0.0161 & 0.3121  & 0.0434 & 0.1204 & 0.0158 \\
DI-SCM & 0.0983 &0.0205 &0.2272 &  0.0477 & 0.0888 & 0.0191 \\
L-SCM& 0.1467 &0.0386 &0.3336 & 0.0874 & 0.1297 & 0.0328 \\
na\"ive Bayes SCM& 0.1109 & 0.0049 &0.2527 & 0.0127 & 0.0992 &0.0045 \\
Bayes SCM& 0.0918 & 0.0039 &0.2109 & 0.0097 & 0.0816 & 0.0037 \\
   \hline
\end{tabular}
\end{table}
  
\section{Empirical application}
To demonstrate the applicability of the proposed Bayesian inference for synthetic control methods, we revisit the case study of \citet{abadie2003economic}, which estimates the effect of the terrorist conflict in the Basque Country on its per-capita GDP. We use the same data as \citet{abadie2003economic} and the data is publicly available in \texttt{R} package \texttt{Synth}. Recently, \cite{chernozhukov2018inference} and \cite{firpo2018synthetic} replicate \citet{abadie2003economic} in the context of developing synthetic control methods with different statistical methods. While the original study of \citet{abadie2003economic} finds about 10 percent points reduction in the per-capita GDP caused by the terrorism activities, the inferences of \cite{chernozhukov2018inference} and \cite{firpo2018synthetic} lead to mixed conclusions on whether the effects are statistically significant or not. 

Similar to \cite{abadie2003economic,chernozhukov2018inference} and \cite{firpo2018synthetic}, we estimate the economic effects of the terrorist conflict in the Basque Country using other 16 Spanish regions as potential control units, i.e. the donor pool. In \citet{abadie2003economic}, per-capita GDP from 1955 to 1997 is used as the outcome and 12 variables associated with economic growth potential are considered as the time-invariant covariates, where the period 1960--1969 is used to construct a synthetic control unit. The result of \citet{abadie2003economic} is reproduced by \texttt{R} package \texttt{Synth} following the instructions in \citet{abadie2011synth}. 

\begin{figure}[t]
    \centering
    \includegraphics[scale=0.5]{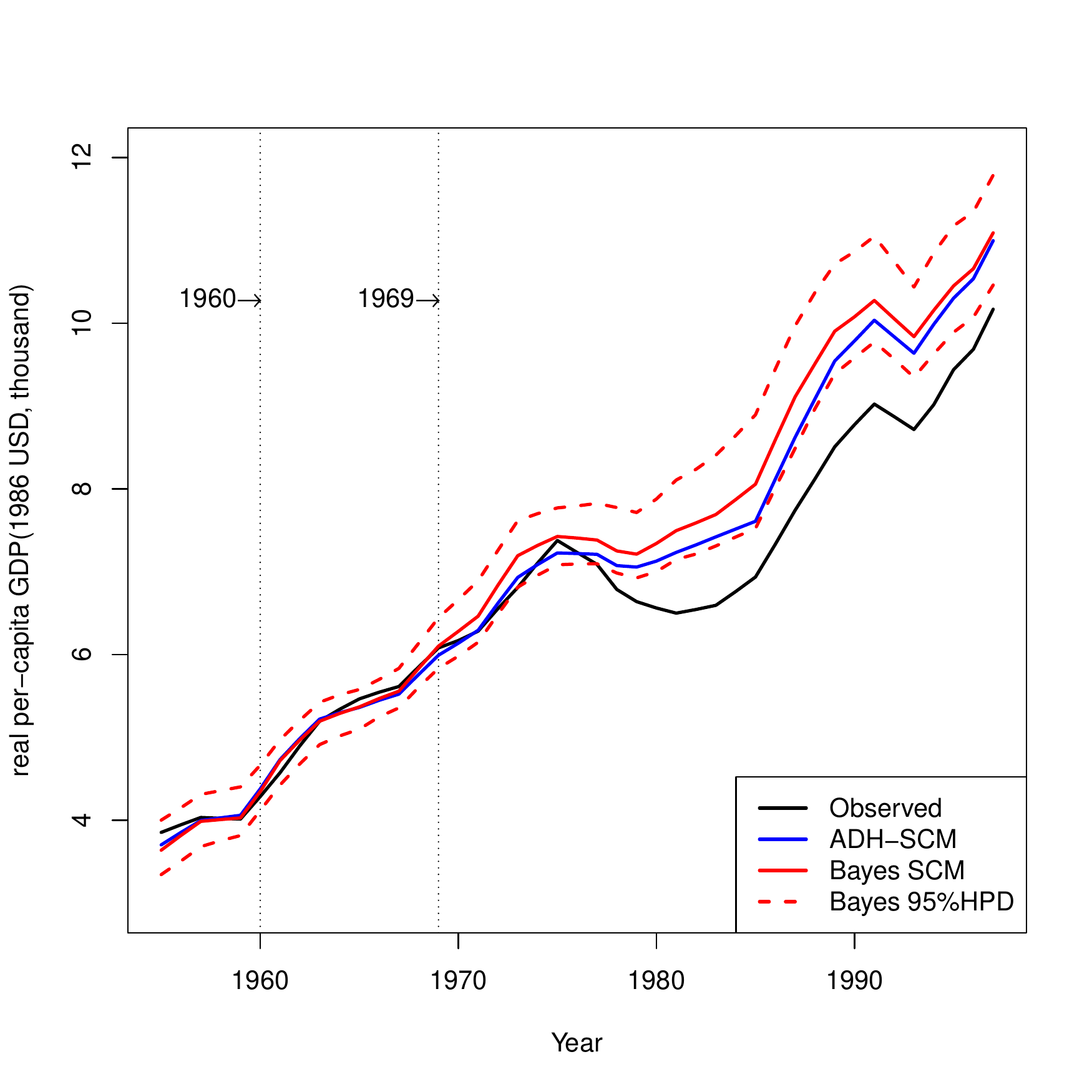}\\
    \caption{\label{fig:1} Trends in per-capita GDP: Basque Country versus synthetic Basque Country.}
\end{figure}

Figure \ref{fig:1} displays the per-capita GDP trajectory of the Basque Country and the counterfactual outcomes obtained by ADH-SCM \citep{abadie2003economic} and Bayes SCM (Proposed) for the 1955--997 period. For Bayes SCM, 95\% highest posterior density (HPD) bands are obtained by using 10,000 consecutive samples
generated from the Gibbs sampler of Section \ref{sec:Gibbs} after a burn-in period of 2,000 iterations. While the actual treatment unit and both synthetic controls show similar patterns during the pre-treatment period (1955-1969), the counterfactuals that are estimated by Bayes SCM and from 1970 the outcomes of Bayes SCM and ADH-SCM are higher than the per capita GDP of the Basque Country. Moreover, the counterfactuals estimated by Bayes SCM are slightly higher than those of ADH-SCM.

\begin{figure}[t]
    \centering
    \includegraphics[scale=0.5]{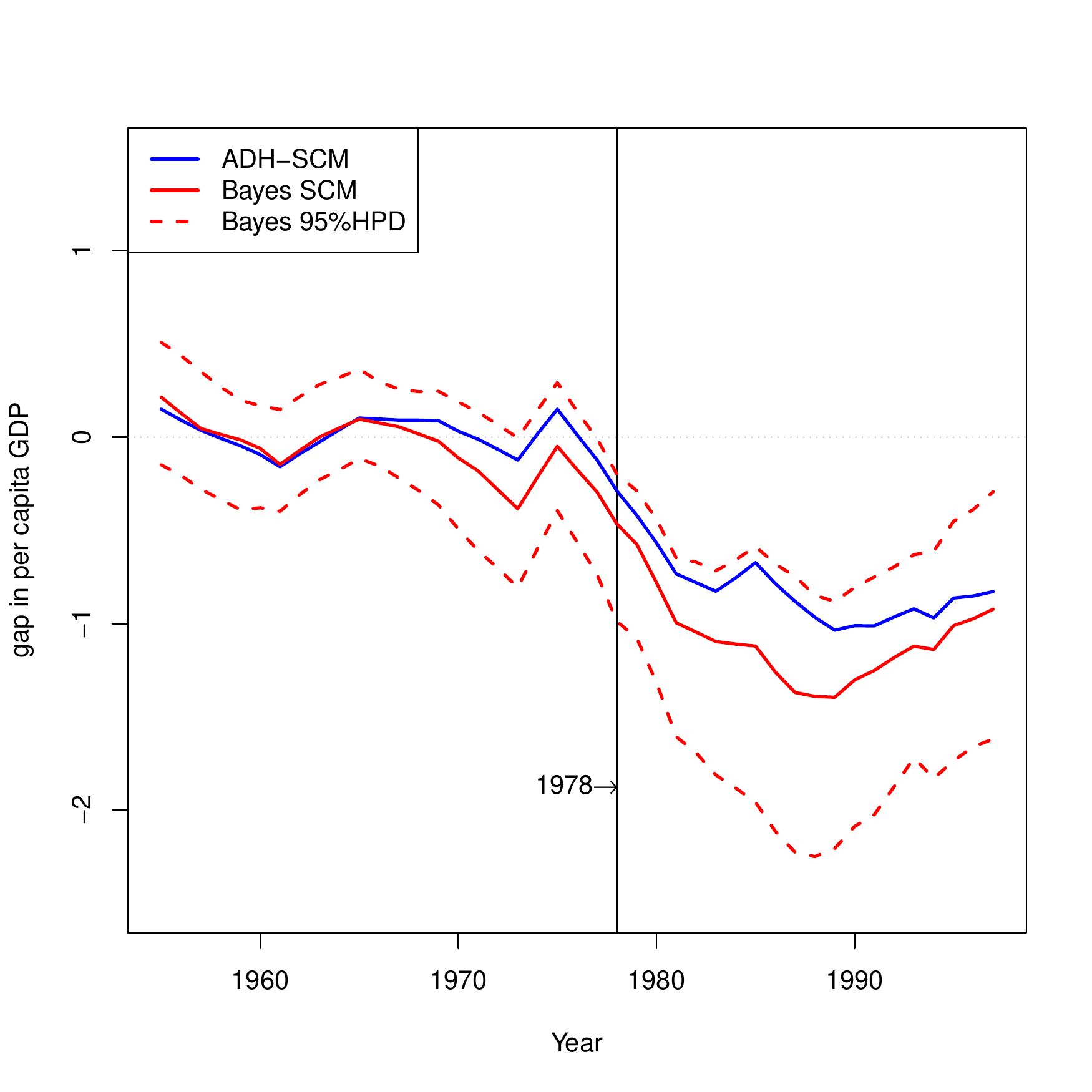}\\
    \caption{\label{fig:2} Estimated gap in per-capita GDP between Basque Country versus synthetic Basque Country}
\end{figure}

The posterior inferences on the treatment effects are shown in Figure \ref{fig:2}. We find that after 1977 the treatment effects become significant, i.e. the reduction in per-capita GDP for the Basque Country is statistically significant during the period 1978--1997, which is the subset of the post-treatment period. 

\begin{table}[t]
\centering
 \caption{\label{Table:2} The average of per-capita GDP gaps between Basque Country and the synthetic Basque Country for the 1978--1997 period}
\begin{tabular}{rrr}
  \hline
 Method & Estimate& 95\% HPD \\ 
  \hline
  Bayes SCM & -1.077 & (-1.722,-0.513) \\ 
  ADH-SCM  & -0.807 & \\ 
   \hline
\end{tabular}
\end{table}

Table \ref{Table:2} reports the average treatment effect during the 1978--1997 period. From the result of Bayes SCM, we find that the terrorist conflict leads to about 13.5 percent loss in per capita GDP in the Basque Country whereas ADH-SCM estimates is about 10 percent decrease in per capita GDP due to the terrorist conflict. 
\begin{table}[t]
\centering 
\caption{\label{Table:3} Estimated weights for the synthetic Basque Country}
\begin{tabular}{rrr}
  \hline
 Region & Bayes SCM & ADH-SCM \\ 
  \hline
(Intercept) & -0.180 & 0.000  \\ 
Andalucia  & 0.000 & 0.000 \\ 
 Aragon & 0.000 & 0.000 \\ 
 Principado De Asturias & 0.000 & 0.000 \\ 
Baleares & 0.278 & 0.000 \\ 
  Canarias& 0.000 & 0.000   \\ 
 Cantabria & 0.000 & 0.000  \\ 
  Castilla Y Leon & 0.000 & 0.000   \\ 
  Castilla-La Mancha & 0.000 & 0.000  \\ 
  Cataluna & 0.343 & 0.851  \\ 
  Comunidad Valenciana & 0.008 & 0.000   \\ 
  Extremadura & 0.000 & 0.000  \\ 
  Galicia & 0.000 & 0.000  \\ 
  Madrid & 0.372 & 0.149  \\ 
  Murcia & 0.000 & 0.000  \\ 
  Navarra & 0.000 & 0.000  \\ 
  Rioja & 0.000 & 0.000 \\ 
   \hline
\end{tabular}
\end{table}

The difference between the outcomes from Bayes SCM and ADH-SCM is explained by the estimated weights. Bayes SCM constructs the synthetic Basque Country using three regions (Baleares, Cataluna, and Madrid) and the intercept term while that of ADH-SCM is based on only two regions (Cataluna and Madrid) as reported in Table \ref{Table:3}. We suspect that the parallel shift plays a substantial role in such difference in the estimated treatment effects. 

Finally, Table \ref{Table:2} also reports the HPD of the ATE during the 1978--1997 period from Bayes SCM estimation. Obviously, given that the TEs are significant as we see in Figure \ref{fig:2}, it shows that the ATE is statistically significant at 95\% level. However, note that we can estimate the ATEs for any sub-periods in the post-treatment periods and state statistical inferences.

\section{Concluding remarks}
As discussed by \cite{athey2017state}, synthetic control methods have been one of the most influential innovations in observational studies. Establishing a clear theoretical foundation on the estimation of the weights for each control unit to create a synthetic control and the statistical inferences is crucial for improving our understanding on how to implement synthetic control approaches. In this paper, we discuss the importance of preserving the convex hull restriction of \cite{abadie2003economic} and \cite{abadie2010synthetic} and propose a synthetic control method that uses a \emph{parallelly shiftable} convex hull to allow some flexibility in estimation while minimizing the risk of extrapolation. We also provide a theoretical discussion on the relationship between the convex hull restriction and the sparsity in the estimation of the weights using the Lagrangian duality.

We then propose a novel Bayesian framework that implements the synthetic control method with the \emph{parallelly shiftable} convex hull and provides a useful Bayesian inference. The simulation results show that our proposed approach leads to smaller MSEs in estimating treatment effects compared to alternative approaches. Further researches in understanding the small sample properties such as the degrees of bias and efficiency of the alternative estimators across different settings such as changes in the number of donors, and the number of predictors, remain. The real data application to the example of \cite{abadie2003economic} finds that the treatment effects are only significant for the subset of the post-treatment period, which previous applications fail to capture.

\end{document}